\documentclass[12pt]{amsart}
\usepackage[left=15mm, right=15mm, top=10mm, bottom=15mm]{geometry}

\usepackage[english]{babel}
\usepackage[T1]{fontenc}
\usepackage[utf8]{inputenc}

\theoremstyle{plain}
    \newtheorem{theorem}{Theorem}[section]
    \newtheorem{lemma}[theorem]{Lemma}

\theoremstyle{definition}

    \newtheorem{example}[theorem]{Example}

\usepackage[
    draft = false,
    unicode = true,
    colorlinks = true,
    allcolors = blue,
    hyperfootnotes = true,
    citecolor = red
]{hyperref}
\usepackage{amsfonts,amsmath,amssymb,amscd,amsthm,url}
\usepackage[inline]{enumitem}
\usepackage{graphicx,epsf,afterpage, wrapfig}
\usepackage{soul}
\usepackage{multicol}
\usepackage{array}
\usepackage{braket}
\usepackage{epigraph}
\usepackage{rotating, floatflt}
\usepackage{xstring} % Пакет для надежной работы со строками

\usepackage[
backend=biber,
style=numeric-comp,
sorting=none,
giveninits=true
]{biblatex}
\usepackage{csquotes}

\usepackage{xcolor}

\usepackage{ulem}

\usepackage{tikz}
\usetikzlibrary{positioning, calc, intersections, through, arrows, matrix, chains, math}
\usetikzlibrary{arrows.meta}
\usetikzlibrary{shadows}
\usetikzlibrary{graphs}
\usetikzlibrary{graphs.standard}
\usetikzlibrary{patterns}
\usetikzlibrary{decorations.pathreplacing,calligraphy,backgrounds}
\DeclareFieldFormat{usera}{\href{https://arxiv.org/abs/#1}{\texttt{arXiv:#1}}}

\renewbibmacro*{finentry}{\printfield{usera}\newunit\finentry}

\renewbibmacro{in:}{%
  \iffieldundef{usera}
    {\printtext{\bibstring{in}\intitlepunct}}
    {}%
}

\newcommand\norm[1]{\ensuremath{\left\lVert#1\right\rVert}}
\newcommand\abs[1]{\ensuremath{\left\lvert#1\right\rvert}}

\DeclareMathOperator{\id}{id}

\DeclareMathOperator{\interior}{int}

\newcommand{\Hcal}{\mathcal{H}}

\newcommand{\Ucal}{\mathcal{U}}

\newcommand{\Abf}{{\ensuremath{\mathbf{A}}}}

\newcommand{\Fbf}{{\ensuremath{\mathbf{F}}}}
\newcommand{\Gbf}{{\ensuremath{\mathbf{G}}}}
\newcommand{\Hbf}{{\ensuremath{\mathbf{H}}}}
\newcommand{\Ibf}{{\ensuremath{\mathbf{I}}}}

\newcommand{\Ubf}{{\ensuremath{\mathbf{U}}}}

\newcommand{\xbf}{{\ensuremath{\mathbf{x}}}}

\newcommand{\zero}{{\ensuremath{\mathbf{0}}}}

\newcommand{\ufrak}{{\ensuremath{\mathfrak{u}}}}

\newcommand{\R}{\ensuremath{\mathbb{R}}}

\newcommand{\Cx}{\ensuremath{\mathbb{C}}}

\newcommand{\N}{\ensuremath{\mathbb{N}}}

\newcommand{\seg}[1]{\ensuremath{[\,#1\,]}}

\renewcommand{\geq}{\geqslant}

\renewcommand{\leq}{\leqslant}

\newcounter{mcnt}

\newcounter{wordcnt}

\begin{document}

\title{Quantum Kolmogorov--Arnold representation theorem \\ for continuous unitary-valued maps}

\author{Sviatoslav V. Dzhenzher}

\begin{abstract}
 The classical Kolmogorov--Arnold representation theorem states that any continuous multivariate function can be exactly decomposed into a finite composition of univariate continuous functions and addition operations.
 This foundational result has recently inspired the development of Kolmogorov--Arnold Networks (KANs) in classical machine learning, as well as their extensions into the quantum domain (QKANs). In this paper, we establish two quantum analogues of the Kolmogorov--Arnold representation theorem for continuous unitary-valued maps of several variables within an open $1$-neighbourhood of the identity matrix \(O_1(\mathbf{I}) \subset \mathcal{U}(n)\).
 First, we prove a representation theorem that yields an exact additive decomposition inside the matrix exponent of anti-Hermitian-valued maps.
 Second, due to the non-commutative nature of quantum operators, we derive a factorised version expressing the target unitary map as a finite sequential product of univariate matrix exponentials. Finally, we provide a concrete topological counterexample based on the lifting property of \(\mathcal{SU}(2)\) to demonstrate that these local representation theorems cannot be globally extended to the entire unitary group \(\mathcal{U}(n)\) without encountering fundamental structural obstructions.
\end{abstract}

\thanks{\hspace{-4mm}
S.\,V. Dzhenzher: sdjenjer@yandex.ru. orcid: 0009-0008-3513-4312
\\
% V.\,Zh. Sakbaev: fumi2003@mail.ru. orcid: 0000-0001-8349-1738
% \\
% V.\,Zh. Sakbaev: Keldysh Institute of Applied Mathematics of Russian Academy of Sciences 125047, Miusskaya pl. 4, Moscow, Russia
% \\
% All authors:
Moscow Institute of Physics and Technology 141701, Institutskii per. 9, Dolgoprudny, Russia}

\maketitle
\thispagestyle{empty}

\noindent \emph{Keywords:} Quantum Kolmogorov--Arnold representation theorem, unitary-valued maps, Lie algebras, quantum machine learning, Kolmogorov--Arnold Networks (KAN), factorisation of unitary evolutions, lifting property.

\vspace{3mm}
\noindent \emph{MSC 2020:} 
Primary: 26B40, %(Real functions of several variables: Representations with sums or products)
41A63; % Multidimensional approximation
Secondary: 22E70, % (Applications of Lie groups to the sciences; explicit representations)
15A16, % (Matrix exponential and other matrix functions),
81P68, % (Quantum computation, quantum information),
68T07. % (Artificial neural networks and deep learning)

% 13P11. % Relations of commutative algebra with algebraic geometry
% 68T07  % Statistical learning theory (including deep learning, artificial neural networks)

\section{Introduction}\label{s:intro}

This line of research originates from Hilbert's 13th problem \cite{wiki-Hilb-13-problem, Morris-2021-hilbert13}.
In 1836, Hamilton showed \cite{Hamilton1836} that any seventh-degree equation can be reduced to an equation of the form
\[
    x^7 + ax^3 + bx^2 + cx + 1 = 0.
\]
Hilbert asked whether the solution \(x=x(a,b,c)\) of this equation, considered as the function of three variables, can be expressed as the superposition of functions of two variables.
Originally, Hilbert posed the question for algebraic functions, but he subsequently extended the conjecture to continuous functions as well.
While the continuous variant was later solved affirmatively, the original algebraic formulation exhibits fundamentally different behaviour due to specific algebraic and topological obstructions (see, for example, an overview in 1976 by Arnold and Shimura \cite{Arnold-Shimura-1976-superposition} on the superposition of algebraic functions). 
The continuous problem was given a positive solution in 1956--57 by Kolmogorov and Arnold \cite{Kolmogorov56, Arnold57, Kolmogorov57}.
In a series of works, they proved the result now known as the Kolmogorov--Arnold (KA) representation theorem (see Theorem~\ref{t:ka}; in computer science, it is sometimes referred to as the Kolmogorov Superposition Theorem).
This theorem states that any continuous multivariate function can be exactly decomposed into a finite composition of univariate continuous functions and addition operations.
Thus, the KA theorem completely settled Hilbert's 13th problem in a significantly more general form.

Despite its theoretical beauty, the classical KA theorem was long criticised for its non-constructive nature and the low smoothness of the resulting inner and outer functions. Specifically, Vitushkin \cite{Vitushkin-1954} and Fridman \cite{Fridman} demonstrated strict limitations regarding the differentiability of these representations, which initially hindered their integration into numerical analysis.
At first, the KA representation theorem had many developments and implications, mostly theoretical.
Ostrand \cite{Ostrand65} generalised it to functions on finite-dimensional compact spaces, and Sternfeld \cite{Sternfeld85, Sternfeld89} proved the lower bound on the dimension of these compact sets. These results bound the number of univariate functions in KA representation; see also \cite{Dzhenzher24} for an overview.
However, the recent computational paradigm shift has completely transformed this perspective, demonstrating that the structural core of the theorem provides an exceptionally powerful framework for function approximation without the curse of dimensionality.
Nearly three decades passed before it was recognised \cite{HechtNielsen87, Kurkova91} that the KA theorem also provides a significant contribution to the topic of neural networks (NNs).
Hecht-Nielsen showed that KA helps to compress any computable function into a three-layer feedforward neural network.
This line of thought regarding bounded-width architectures was further pushed in 1999 by Maiorov and Pinkus \cite{Maiorov-Pinkus}, who demonstrated that an MLP with just two hidden layers and a fixed, bounded number of units can act as a universal approximator, provided a specific analytic sigmoidal activation function is chosen.

These developments opened the door to the practical applications of KA.
The approach was revitalised completely in 2024 by Kolmogorov--Arnold Networks (KANs) \cite{KAN}, which replace fixed activation functions with learnable spline-based edge activations, offering a more interpretable alternative to MLPs.
This framework was further extended in \cite{KAN2} into scientific discovery and in \cite{CKAN} into grid-structured data.
The progression of these ideas to the quantum domain has led to Quantum Kolmogorov--Arnold Networks (QKANs).
For example:
\begin{itemize}    
    \item \cite{QKAN-variation} introduces trainable quantum circuits with data re-uploading to represent edge activations,
    % There, the unitary evolutions are parameterised as in our paper.
    
    \item \cite{QKAN-enh-var} uses a tiling technique for improved noise robustness on Noisy Intermediate-Scale Quantum (NISQ) devices,

    \item \cite{Werner-2025} adapts classical KAN architectures to quantum forms using a Quantum Circuit Born Machine (QCBM),

    \item \cite{QKAN-Ivashkov-etal} utilises block-encodings and Quantum Singular Value Transformation for rigorous quantum linear algebra, enabling efficient state preparation.

    \item \cite{QKAN-china} leverages parameterised quantum circuits to enable learnable, basis-free activation functions for improved function approximation in quantum machine learning.
\end{itemize}

Beyond network architectures, KANs have proven effective in physics applications.
Shamim et al. \cite{QKAN-spins} employed KAN-based wavefunctions for studying quantum spin systems with higher efficiency than traditional neural networks.
Sen et al. \cite{QKAN-time-series} introduced physics-informed time series analysis using KANs with Ehrenfest constraints.
Furthermore, these architectures have been successfully merged \cite{QCPIKAN} with quantum computing to solve partial differential equations, as demonstrated by Quantum-Classical Physics-Informed Kolmogorov-Arnold Networks (QCPIKAN).
Finally, in \cite{Dzhenzher26-alg-QKA}, a theoretical foundation was provided by establishing an algebraic Kolmogorov--Arnold representation theorem for quantum measurements, showing that unentangled state estimation can be decomposed through local observables.
There, questions arose regarding the establishment of the Quantum Kolmogorov--Arnold (QKA) representation theorem for unitary evolutions, or, more generally, for quantum channels.
The extension of KANs to the quantum domain promises to bridge the gap between classical structural representations and quantum circuit architectures. Yet, a rigorous mathematical framework for representing continuous unitary-valued maps via parameterised paths in Lie groups remains underdeveloped. Motivated by this gap, this paper establishes a formal bridge between classical superposition theory and unitary evolutions. By mapping multivariate parameters into the Lie algebra \(\mathfrak{u}(n)\), we present exact decomposition strategies that provide a solid theoretical foundation for designing structurally sound QKANs.

In this paper, we address those questions from \cite{Dzhenzher26-alg-QKA} while considering the continuous unitary-valued maps of several variables.
Specifically, we consider unitary evolutions parameterised by \(\xbf\in\seg{0,1}^d\);
this approach is close to \cite{QKAN-variation, QKAN-enh-var, QKAN-Ivashkov-etal, QKAN-china, QCPIKAN}.
Instead of sums of functions, we consider compositions of matrix exponentials of anti-Hermitian-valued maps.
Since the operators may not commute, we consider two variants of QKA: for the exponent of the sum of operators (Theorem~\ref{t:qka-repr}), and for the factorised exponentials (Theorem~\ref{t:qka-fact}).
Note that even in classical structures, one can introduce non-commutativity via the Poisson bracket \cite{Brunton-2022, Morgan-2020}.

The structure of the article is as follows.
Section~\ref{s:main} is devoted to the necessary background and main results (Theorem~\ref{t:qka-repr} and Theorem~\ref{t:qka-fact}).
In Section~\ref{s:proof-repr}, we prove Theorem~\ref{t:qka-repr}.
In Section~\ref{s:proof-fact}, we prove Theorem~\ref{t:qka-fact}.
Section~\ref{s:conclusion} is devoted to the conclusions and discussions of the results.

\section{Background and main results}\label{s:main}

By \(\xbf\in X^d\) we denote a $d$-dimensional tuple \(\xbf=(x_1,\ldots,x_d)\) with \(x_j\in X\).

Let us recall the Kolmogorov--Arnold representation theorem \cite{Kolmogorov56, Arnold57, Kolmogorov57}.
The proof (and the formulation) has been refined over the years, for example, in \cite{Sprecher1965, Kahane, Lorentz-Golitschek-Makovoz}.
We are using the modified version of this theorem from \cite{Hedberg-appx, Brattka2007}.

\begin{theorem}[Kolmogorov--Arnold representation; 1956--57]\label{t:ka}
    Let $d > 1$ be an integer.
    Then there exist $2d+1$ continuous ``inner'' functions \(\phi_1,\ldots,\phi_{2d+1}\colon\seg{0,1}^d\to\R\) of the form
    \[
        \phi_j(\xbf) = \sum_{k=1}^d \phi_{j,k}(x_k),
    \]
    such that for any continuous ``target'' function \(f\colon\seg{0,1}^d\to\R\), there exists a uniformly continuous ``outer'' function \(g\colon\R\to\R\), such that for any \(\xbf\in\seg{0,1}^d\),
    \[
        f(\xbf) = \sum_{j=1}^{2d+1} g(\phi_j(\xbf)).
    \]
\end{theorem}

In 1965, Sprecher showed \cite{Sprecher1965} that \(\phi_{j,k}\) in Theorem~\ref{t:ka} can be chosen of the form \(\phi_{j,k}(x) = \phi_{j,1}(x)\gamma_k\), where \(\gamma_1,\ldots,\gamma_d\) are rationally independent.
For example, \(\gamma_1,\ldots,\gamma_d\) can be the square roots of pairwise distinct prime numbers.
Vitushkin \cite{Vitushkin-1954} showed that the theorem is false if we require all functions \(f,g,\phi_j\) to be continuously differentiable; in fact, the set of tuples \((\phi_j)\) corresponding to a function $f$ becomes nowhere dense, which spoils the application of the Baire category theorem, firstly introduced by Kahane \cite{Kahane}.
Fridman \cite{Fridman} showed that the ``inner'' functions \(\phi_j\) can be chosen $1$-Lipschitz.

We identify bounded linear operators on \(\Cx^n\) with complex \(n\times n\) matrices.
Denote by \(\Ucal(n)\) the Lie group of unitary \(n\times n\)-matrices.
Denote by \(\ufrak(n)\) the corresponding Lie algebra of anti-Hermitian \(n\times n\)-matrices.
Let us denote by
\[
    O_1(\Ibf) := \left\{\Ubf\in\Ucal(n) \;:\; \norm{\Ubf-\Ibf}<1\right\} \subset\Ucal(n)
\]
the open $1$-neighbourhood of the identity matrix $\Ibf$.
The definition of \(O_1(\Ibf)\) depends on the choice of the matrix norm, but our results remain independent of it.
The choice of norm affects only the constants, with its sole requirement being
\[
    O_1(\Ibf) \subsetneq\Ucal(n).
\]

\begin{theorem}[Unitary representation KA]\label{t:qka-repr}
    Let $n\geq 1$ be an integer, and let \(d>1\) be the dimension of the parameter space.
    Then there exist \(m=n^2(2d+1)\) fixed ``inner'' anti-Hermitian operators \(\Hbf_1,\ldots,\Hbf_m\in\ufrak(n)\) and fixed ``inner'' continuous functions \(\phi_1,\ldots,\phi_m\colon [\,0,1\,]^{d}\to \R\) of the form
    \[
        \phi_j(\xbf) = \sum_{k=1}^d \phi_{j,k}(x_k),
    \]
    such that for any ``target'' continuous unitary-valued map \(\Ubf\colon[\,0,1\,]^d\to O_1(\Ibf)\subset\Ucal({n})\), there exist ``outer'' uniformly continuous functions $g_1,\ldots,g_m\colon\R\to\R$ such that
    \[
        \Ubf(\xbf) = e^{\sum_{j=1}^mg_j(\phi_j(\xbf))\Hbf_j}.
    \]
    Moreover, the expression in the exponent can be represented in the form
    \[
        \sum_{k=1}^{n^2}\sum_{s=1}^{2d+1}g_k(\phi_s(\xbf))\Hbf_k.
    \]
\end{theorem}

From the physical point of view, it would be more natural to split the exponent of the sum into the product of exponentials.
This would map the non-trivial entangled exponentials onto a sequential application of unitary evolutions with fixed operators \(\Hbf_j\) and configurable weights \(g_j(\phi_j(\xbf))\).
This naive splitting is blocked by the possible non-commutativity of the operators \(\Hbf_j\).
Therefore, we present the following factorised version of Theorem~\ref{t:qka-repr}.

\begin{theorem}[Unitary factorisation KA]\label{t:qka-fact}
    Let $n\geq 1$ be an integer, and let \(d>1\) be the dimension of the parameter space.
    Then there exist \(m=m(n,d)\) fixed ``inner'' anti-Hermitian operators \(\Hbf_1,\ldots,\Hbf_m\in\ufrak(n)\) and fixed ``inner'' continuous functions \(\phi_1,\ldots,\phi_m\colon [\,0,1\,]^{d}\to \R\) of the form
    \[
        \phi_j(\xbf) = \sum_{k=1}^d \phi_{j,k}(x_k),
    \]
    such that for any ``target'' continuous unitary-valued map \(\Ubf\colon[\,0,1\,]^d\to O_1(\Ibf)\subset\Ucal({n})\), there exist ``outer'' uniformly continuous functions $g_1,\ldots,g_m\colon\R\to\R$ such that
    \[
        \Ubf(\xbf) = \prod_{j=1}^m e^{g_j(\phi_j(\xbf))\Hbf_j}.
    \]
    Moreover, the product of the exponentials can be represented in the form
    \[
        \prod_{k=1}^M\prod_{s=1}^{2d+1} e^{g_k(\phi_s(\xbf))\Hbf_k}.
    \]
    where \(M=M(n) = \frac{m}{2d+1}\).
\end{theorem}

Proofs of both Theorem~\ref{t:qka-repr} and Theorem~\ref{t:qka-fact} heavily rely on Lemma~\ref{l:lifting} (Lifting), for which the usage of the neighbourhood \(O_1(\Ibf)\) is crucial.
Moreover, Theorem~\ref{t:qka-repr}, as a stronger version of Lemma~\ref{l:lifting}, does not hold for the whole group \(\Ucal(n)\) instead of \(O_1(\Ibf)\).
It is hard to tell whether Theorem~\ref{t:qka-fact} holds for \(\Ucal(n)\), but it probably does not.
See some discussion on this after the proof of Lemma~\ref{l:lifting}, and in Example~\ref{ex:bad-lift}.

\section{Proof of Theorem~\ref{t:qka-repr} (Unitary representation)}\label{s:proof-repr}

\begin{lemma}[Lifting]\label{l:lifting}
    Let $n\geq 1$ be an integer, and let \(d\geq 1\) be the dimension of the parameter space.
    Let \(\Ubf\colon[\,0,1\,]^d\to O_1(\Ibf)\subset\Ucal(n)\) be a continuous unitary-valued map.
    Then there exists a continuous anti-Hermitian-valued map \(\Abf\colon\seg{0,1}^d\to\ufrak(n)\), such that
    \[
        \Ubf(\xbf) = e^{\Abf(\xbf)}.
    \]
    Moreover, there exists a universal constant $C$ (independent of $\Ubf,n,d,\xbf$) such that \(\norm{\Abf(\xbf)}\leq C\).
\end{lemma}

\begin{proof}
    Since \(\norm{\Ubf(\xbf)-\Ibf} < 1\), by \cite[Theorem~11.5.1]{Bernstein2009matrix} we have the principal branch logarithm operator
    \[
        \Abf(\xbf) := \sum_{j=1}^\infty \frac{(-1)^{j+1}}{j}(\Ubf(\xbf)-\Ibf)^j.
    \]
    It is continuous since the series converges uniformly on a compact cube \(\seg{0,1}^d\).

    In order to specify $C$, we need to agree on our choice of the matrix norm.
    Let us use the spectral norm.
    Then, since \(\Abf(\xbf)\) is anti-Hermitian, for each of its eigenvalues \(\lambda\in i\R\) we have
    \[
        \abs{e^{\lambda}-1} < 1.
    \]
    Since the above logarithm series yields the principal branch, we obtain that \(\abs{\lambda}\leq\pi\), and so
    \[
        \abs{\lambda} < \frac{\pi}{3}.
    \]
    Hence,
    \[
        \norm{\Abf(x)}_2 < \frac{\pi}{3},
    \]
    and thus \(C=\frac{\pi}{3}\) for the spectral norm.
\end{proof}

Note that Lemma~\ref{l:lifting} does not hold for continuous maps \(\Ubf\colon\seg{0,1}^d\to\Ucal(n)\).
See Example~\ref{ex:bad-lift}, where it is explicitly shown for $n=2$ and for \(\mathcal{SU}\subset\Ucal\). There, the closed ball \(B^3\) is taken instead of \(\seg{0,1}^3\) just for simplicity, since they are homeomorphic.

\begin{proof}[Proof of Theorem~\ref{t:qka-repr}]
    Fix the basis \(\Fbf_1,\ldots,\Fbf_N \in \ufrak(n)\) of anti-Hermitian matrices, where \(N=n^2\).
    Applying Lemma~\ref{l:lifting} (Lifting), we obtain the continuous anti-Hermitian-valued map \(\Abf=\Abf(\xbf)\).
    Then
    \[
        \Abf(\xbf) = \sum_{k=1}^N f_k(\xbf) \Fbf_k,
    \]
    where \(f_1,\ldots,f_N\colon[\,0,1\,]^d\to\R\) are continuous functions.

    Now fix \(2d+1\) universal ``inner'' functions \(\phi_j\) from Theorem~\ref{t:ka} (Kolmogorov--Arnold representation). For each $k\in\{1,\ldots,N\}$ take a uniformly continuous function \(g_k\) such that
    \[
        f_k(\xbf) = \sum_{j=1}^{2d+1} g_{k}(\phi_j(\xbf)).
    \]
    It remains to replicate each \(\Fbf_k\) and \(g_k\) exactly \(2d+1\) times, and each \(\phi_j\) exactly $N$ times.
\end{proof}

\section{Proof of Theorem~\ref{t:qka-fact} (Unitary factorisation)}\label{s:proof-fact}

\begin{lemma}[Factorisation]\label{l:fact}
    Let $n\geq 1$ be an integer, and let \(d\geq1\) be the dimension of the parameter space.
    Then there exist \(M=M(n)\) anti-Hermitian operators \(\Gbf_1,\ldots,\Gbf_M \in\ufrak(n)\) such that
    for any continuous unitary-valued map \(\Ubf\colon[\,0,1\,]^d\to O_1(\Ibf)\subset\Ucal(n)\),
    there exist continuous functions \(\theta_1,\ldots,\theta_M\colon[\,0,1\,]^d\to\R\) such that
    \[
        \Ubf(\xbf) = \prod_{j=1}^Me^{\theta_j(\xbf)\Gbf_j}.
    \]
\end{lemma}

\begin{proof}
    Fix the basis \(\Fbf_1,\ldots,\Fbf_N \in \ufrak(n)\) of anti-Hermitian matrices, where \(N=n^2\).
    Define the function \(F\colon\R^N\to \Ucal(n)\) by
    \[
        F(\theta_1,\ldots,\theta_N) := \prod_{s=1}^Ne^{\Fbf_s\theta_s}.
    \]
    Its differential \(dF_\zero\colon T_\zero\R^N\to T_\Ibf\Ucal(n)\cong\ufrak(n)\) equals
    \[
        dF_\zero(\theta_1,\ldots,\theta_N) = \sum_{s=1}^N \theta_s\Fbf_s.
    \]
    Since \(\Fbf_1,\ldots,\Fbf_N\) form a basis, the differential is a linear isomorphism.
    Hence, by the Inverse Function Theorem (see, for example, \cite[Theorem~C.34]{Lee2012Smooth}), there exist open neighbourhoods \(V\supset \zero\in\R^N\) and \(W\supset\Ibf\in\Ucal(n)\) such that \(F\colon V\to W\) is bijective.
        
    Apply Lemma~\ref{l:lifting} (Lifting) in order to obtain a continuous parameterisation \(\Abf(\xbf)\) such that
    \[
        e^{\Abf(\xbf)} = \Ubf(\xbf),
        \qquad \norm{\Abf(\xbf)} \leq C.
    \]
    Choose a sufficiently large \(K\in\N\) such that \(e^{\Abf(\xbf)/K}\in W\) for all \(\xbf\in\seg{0,1}^d\).
    Note that the choice of $K$ depends only on the choice of basis vectors \(\Fbf_s\) and the universal constant $C$.
    Now we have $N$ continuous functions \(\theta_1,\ldots,\theta_N\colon\seg{0,1}^d\to\R\) such that
    \[
        e^{\Abf(\xbf)/K} = \prod_{s=1}^Ne^{\theta_s(\xbf)\Fbf_s}.
    \]
    Since
    \[
        \Ubf(\xbf) = \left(e^{\Abf(\xbf)/K}\right)^K,
    \]
    it remains to replicate functions \(\theta_s\) and \(\Fbf_s\) exactly $K$ times, in order to obtain the required representation for $M=KN$ exponentials.
\end{proof}

\begin{proof}[Proof of Theorem~\ref{t:qka-fact}]
    Fix anti-Hermitian matrices \(\Gbf_1,\ldots,\Gbf_M \in \ufrak(n)\) from Lemma~\ref{l:fact} (Factorisation).
    Then for any \(\Ubf\colon\seg{0,1}^d\to O_1(\Ibf)\subset\Ucal(n)\) we have the representation
    \[
        \Ubf(\xbf) = \prod_{k=1}^Me^{\theta_k(\xbf)\Gbf_k},
    \]
    where \(\theta_k\colon\seg{0,1}^d\to\R\) are continuous functions.
    Now fix \(2d+1\) universal ``inner'' functions \(\phi_j\) from Theorem~\ref{t:ka} (Kolmogorov--Arnold representation).
    For each $k\in\{1,\ldots,M\}$, take a uniformly continuous function \(g_k\) such that
    \[
        \theta_k(\xbf) = \sum_{j=1}^{2d+1} g_{k}(\phi_j(\xbf)).
    \]
    It remains to set \(m=(2d+1)M\) in order to replicate each \(\Gbf_k\) and \(g_k\) exactly \(2d+1\) times, and each \(\phi_j\) exactly $M$ times.
\end{proof}

\section{Conclusions and discussion}\label{s:conclusion}

In this work, we have successfully established two distinct quantum analogues of the KA representation theorem tailored for continuous unitary-valued maps within the local neighbourhood \(O_1(\Ibf) \subset \Ucal(n)\).
Theorem~\ref{t:qka-repr} provides an exact additive representation inside the matrix exponent, effectively compressing multivariate unitary evolutions into a finite sum of univariate fields over fixed anti-Hermitian bases.
Conversely, Theorem~\ref{t:qka-fact} offers a factorised sequential product formulation, which directly aligns with the operational paradigm of quantum gate synthesis and multi-layer quantum circuits.
The exposed results have some further developments that are interesting from both mathematical and physical points of view.

Probably, as a theoretical refinement, one can apply the classical proof of KA (for example, from \cite{Lorentz-Golitschek-Makovoz}) to Theorem~\ref{t:qka-repr} and Theorem~\ref{t:qka-fact}, in order to obtain analogous results for \(g_j=g\) being the same.
Physically, that would mean that we need to prepare only one ``outer'' function $g$.

The ideas of the constructive proof of the Kolmogorov--Arnold representation theorem were given in \cite{Sprecher96, Sprecher97}, and completed in \cite{Braun-Griebel-2009}.
Presumably, leveraging their construction of the ``inner'' functions \(\phi_j\), it is possible to \emph{constructively} prove Theorem~\ref{t:qka-repr} and Theorem~\ref{t:qka-fact}. This may provide an approximative algorithm for the reproduction of unitary evolutions by sequential applications of exponentials of fixed anti-Hermitian operators.

Furthermore, it is worth noting that while the classical KA representation often yields highly non-smooth outer functions, in 2021 Schmidt-Hieber \cite{SchmidtHieber2021} demonstrated that modifications of the theorem can successfully transfer the smoothness properties of the target function to the outer layer, allowing efficient approximations via deep ReLU networks.
Investigating such smoothness transfers in the quantum domain remains an open challenge.
On the other hand, in \cite{Ismailov-2026}, the analogues of KA were proved for discontinuous functions.
Presumably, using these results, it is possible to obtain the analogues of theorems~\ref{t:qka-repr} and~\ref{t:qka-fact} for discontinuous unitary-valued maps.

In \cite{DzhenzherFreedman-25}, the stability of KA was established under reparameterisations of the ``inner'' layer of \(\phi_j\).
It was shown there, that if we know a countable family of adversarial homeomorphisms \(\R^{2d+1}\to\R^{2d+1}\) acting on the ``inner'' functions,
then we can adjust the ``outer'' function $g$ so that the adversarial reparameterisation will not compromise the ``target'' function.
More formally, the theorem is as follows.

\begin{theorem}[Dzhenzher, Freedman; 2025]
    Let $d>1$ be an integer.
    Let \(\mathcal{H}\) be a countable set of homeomorphisms \(\R^{2d+1}\to\R^{2d+1}\).
    Let \(\gamma_1,\ldots,\gamma_d\in\R\) be rationally independent.
    Then, there exists a fixed ``inner'' tuple \(\phi\colon \seg{0,1}\to \seg{0,1}^{2d+1}\) such that for any ``target'' continuous function \(f\colon \seg{0,1}^d\to\R\) and any adversarial homeomorphism \(h\in\mathcal{H}\),
    there exists a uniformly continuous function \(g\colon\R\to\R\) such that for any \(\xbf\in\seg{0,1}^d\),
    \[
        f(\xbf) = \sum_{j=1}^{2d+1} g\Bigl(\gamma_1h_j\bigl(\phi(x_1)\bigr) +\ldots+ \gamma_dh_j\bigl(\phi(x_d)\bigr)\Bigr).
    \]
\end{theorem}

For \(\Hcal=\{\id\}\), this theorem is just the classical KA.
It is clear, based on the proofs above, that analogous results hold for adversarial reparameterisations, acting in Theorem~\ref{t:qka-repr} and Theorem~\ref{t:qka-fact}.
We give below an example of the reformulation of Theorem~\ref{t:qka-fact} (Unitary factorisation).

\begin{theorem}\label{t:qka-stability}
    Let $n\geq 1$ be an integer, and let \(d>1\) be the dimension of the parameter space.
    Let \(\mathcal{H}\) be a countable set of homeomorphisms \(\R^{2d+1}\to\R^{2d+1}\).
    Let \(\gamma_1,\ldots,\gamma_d\in\R\) be rationally independent.
    Then there exist \(M=M(n)\) fixed ``inner'' anti-Hermitian operators \(\Gbf_1,\ldots,\Gbf_M\in\ufrak(n)\) and a fixed ``inner'' continuous tuple \(\phi\colon \seg{0,1}\to \seg{0,1}^{2d+1}\) such that for any ``target'' continuous unitary-valued map \(\Ubf\colon[\,0,1\,]^d\to O_1(\Ibf)\subset\Ucal({n})\) and any adversarial homeomorphism \(h\in\mathcal{H}\),
    there exist uniformly continuous functions \(g_1,\ldots,g_M\colon\R\to\R\) such that for any \(\xbf\in\seg{0,1}^d\),
    \[
        \Ubf(\xbf) = \prod_{k=1}^M\prod_{j=1}^{2d+1} e^{g_k\Bigl(\gamma_1h_j\bigl(\phi(x_1)\bigr) +\ldots+ \gamma_dh_j\bigl(\phi(x_d)\bigr)\Bigr)\Gbf_k}.
    \]
\end{theorem}

It is natural to explore Theorem~\ref{t:qka-repr} and Theorem~\ref{t:qka-fact} on the stability under adversarial perturbations acting on the fixed ``inner'' anti-Hermitian operators.
Probably, using the classical technique of KA, one can prove such analogous results.
A similar result was established in \cite{Dzhenzher26-alg-QKA}, where the stability under adversarial attacks acting both on the fixed ``inner'' observables and on the target observable was shown.
From a practical perspective, our factorisation results (Theorem~\ref{t:qka-fact} and Theorem~\ref{t:qka-stability}) suggest a systematic pipeline for compiling complex, parameter-dependent quantum protocols into sequential single-parameter rotations.
This structural decomposition can significantly simplify the optimisation landscape for variational quantum algorithms, as the training process can be distributed across decoupled univariate spline-based edge activations rather than highly non-linear multivariate cost functions.
Future research may focus on quantifying the exact approximation errors when transitioning from these continuous representations to discrete, finite-depth quantum circuits on noisy intermediate-scale quantum (NISQ) hardware.

While these structural advantages highlight the potential of our framework for scalable quantum network architectures, it is crucial to establish the precise mathematical boundaries of the proven representation theorems.
Specifically, both Theorem~\ref{t:qka-repr} and Theorem~\ref{t:qka-fact} strictly rely on the local geometric structure of the $1$-neighbourhood \(O_1(\mathbf{I})\). A natural question arises as to whether these results can be globally extended to the entire unitary group \(\Ucal(n)\).
To demonstrate that such a global generalisation encounters fundamental topological obstructions, we conclude this paper with a concrete counterexample involving the lifting property on compact domains.
Intriguingly, we have not succeeded in finding the reference for this counterexample in the topology literature, although the example is simple and rather educational.
Perhaps it is considered topological folklore.

\begin{example}\label{ex:bad-lift}
    Here we show that the exponential map \(\exp\colon\mathfrak{su}(2)\to\mathcal{SU}(2)\) does not satisfy the lifting property for a closed ball \(B^3 = \{\vec{x}\in\R^3:\abs{\vec{x}}\leq 1\}\).
    % This example is from a private conversation with Andrey Ryabichev.
    See an overview of smooth manifolds, Lie groups and liftings, for example, in \cite{Hatcher2001AT, Lee2012Smooth}.

    Identify the Lie algebra \(\mathfrak{su}(2)\) with \(\R^3\), and the Lie group \(\mathcal{SU}(2)\) with the unit sphere \(S^3\subset\R^4\).
    The identity matrix \(\Ibf\in\mathcal{SU}(2)\) corresponds to the point \((1,0,0,0)\in S^3\).
    The exponential map is given by
    \[
        \exp\colon \mathfrak{su}(2) \cong \R^3\ni\vec{x} \mapsto \left( \cos \abs{\vec{x}}, \frac{\vec{x}}{\abs{\vec{x}}} \sin \abs{\vec{x}} \right) \in S^3 \cong \mathcal{SU}(2),
        \qquad \exp(0,0,0) = (1, 0, 0, 0) = \Ibf.
    \]
    It is clear that the pre-image \(\exp^{-1}(-\Ibf)\) is equal to the union of spheres
    \[
         S^2_k = \{ \vec{x} \in \R^3 \,:\, \cos\abs{\vec{x}}=-1,\,\sin\abs{\vec{x}}=0 \} = \{ \vec{x} \in \R^3 \,:\, \abs{\vec{x}} = (2k+1)\pi \} \subset \R^3
    \]
    for \(k\in\N_0\).

    Define the smooth map \(f\colon B^3\to S^3\) by
    \[
        f(\vec x) := -\exp(-\pi\vec{x}).
    \]
    This map sends the origin to \(-\Ibf\), and sends the entire boundary \(\partial B^3\) to \(\Ibf\).
    Moreover, $f$ is injective on the interior \(\interior B^3 = B^3\setminus \partial B^3\) of $B^3$.
    Indeed, if \(f(\vec{x})=f(\vec{y})\) for some \(\vec{x},\vec{y}\in \interior B^3\), then \(\cos{\pi\abs{\vec{x}}} = \cos{\pi\abs{\vec{y}}}\).
    This implies \(\abs{\vec{x}}=\abs{\vec{y}}\).
    Also,
    \[
        \frac{\vec{x}}{\abs{\vec{x}}} \sin \abs{\vec{x}} =
        \frac{\vec{y}}{\abs{\vec{y}}} \sin \abs{\vec{y}},
    \]
    which together with \(\abs{\vec{x}}=\abs{\vec{y}}\) implies that \(\vec{x}=\vec{y}\).
    So, $f$ indeed is injective on \(\interior B^3\).
    
    In this paragraph, we show that $f$ admits no continuous lift.
    On the contrary, suppose that there exists a continuous lift \(\tilde{f}\colon B^3\to\mathfrak{su}(2)\cong\R^3\) such that \(\exp\circ\tilde{f}=f\).
    Since \(\exp(\tilde{f}(0,0,0))=f(0,0,0)=-\Ibf\), it follows that \(\tilde{f}(0,0,0)\) lies on some sphere $S^2_k$.
    Moreover, since $f$ is injective on $\interior B^3$, it follows that the lift \(\tilde{f}\) is also injective on \(\interior B^3\).
    Hence \(\tilde{f}(\interior B^3)\) is open, by the Brouwer Invariance of Domain theorem (see, for example, \cite[Theorem~2B.3]{Hatcher2001AT}).
    Therefore, there exists \(\vec{x}\in \interior B^3\setminus\{(0,0,0)\}\) such that \(\tilde{f}(\vec{x})\in S^2_k\).
    For this point,
    \[
        f(\vec{x}) = \exp(\tilde{f}(\vec{x})) = -\Ibf.
    \]
    Hence, $\vec{x}$ must be the origin.
    This leads to the contradiction.
\end{example}

\section*{Data availability}

The research has no associated data.

\section*{Conflict of interest}

The author confirms that he has no conflicts of interest.

\section*{Acknowledgements}

The author is grateful to Andrey Ryabichev for suggesting the core idea for the counterexample presented in Example~\ref{ex:bad-lift}.

\printbibliography
% % \printbibitembibliography

\end{document}